\documentclass[letterpaper, 10 pt, conference]{ieeeconf}  % Comment this line out
\IEEEoverridecommandlockouts                              % This command is only
\usepackage{graphicx}
\usepackage{amsmath}
\usepackage{amssymb}
\usepackage{comment}

\usepackage[ruled]{algorithm}
\usepackage{algpseudocode}
\usepackage{hyperref}
\usepackage{standalone}
\usepackage{tikz}
\usetikzlibrary{automata, positioning, arrows.meta, fit, backgrounds}
\tikzset{every edge/.append style={font=\Large},every path/.append style={font=\Large}, }
\usepackage{circuitikz}
\usepackage{siunitx}
  
\usepackage{amsthm}
\algnewcommand\algorithmicinput{\textbf{Input:}}
\algnewcommand\INPUT{\item[\algorithmicinput]}
\newtheorem{proposition}{Proposition}
\newtheorem{definition}{Definition}
\newtheorem*{problem}{Problem}

\theoremstyle{remark}
\newtheorem*{remark}{Remark}

\title{\LARGE \bf
Verification of Autonomous Systems with Optimal Controllers
}

\author{Dylan Le\(^{\dagger}\)%
\thanks{This work has been submitted to the IEEE for possible publication. Copyright may be transferred without notice, after which this version may no longer be accessible.}%
\thanks{
The authors are with the Department of Computer Science, Rensselaer Polytechnic Institute, 110 8th St, Troy, NY 12180, USA},
Joel McCandless\(^{\dagger}\)\thanks{\(\dagger\) These authors contributed equally to this work.},
Carlos A. Varela,
Radoslav Ivanov% <-this % stops a space
}

\usepackage{letltxmacro}
\usepackage{subcaption}
\usepackage{graphicx}

\LetLtxMacro\Oldfootnote\footnote

\newcommand{\EnableFootNotes}{%
  \LetLtxMacro\footnote\Oldfootnote%
}

\newcommand{\DisableFootNotes}{%
  \renewcommand{\footnote}[2][]{\relax}
}
\begin{document}
\DisableFootNotes
\EnableFootNotes

\maketitle
\thispagestyle{empty}
\pagestyle{empty}

%%%%%%%%%%%%%%%%%%%%%%%%%%%%%%%%%%%%%%%%%%%%%%%%%%%%%%%%%%%%%%%%%%%%%%%%%%%%%%%%
\begin{abstract}
This paper considers the problem of reachability analysis of control systems with
optimal controllers, as a first step towards verifying the safety and correctness
of such systems. Despite their appeal in guaranteeing task satisfaction through
cost minimization, optimal controllers are often challenging to assure. In
particular, as system dynamics grow in complexity, solving the resulting
optimization problem may be difficult, especially given time and computation
constraints on real platforms. Thus, it is essential to verify that, even if the
optimal solution is not always found, such controllers still accomplish the
high-level control objective. In this paper, we focus on gradient descent
algorithms and design a reachability algorithm by treating gradient descent as a
separate (digital) dynamical system, embedded in the original (physical) dynamical
system, with controls as part of the state. We evaluate the feasibility of the
proposed method on two control systems, a two-dimensional quadrotor and a cartpole.
\end{abstract}

%%%%%%%%%%%%%%%%%%%%%%%%%%%%%%%%%%%%%%%%%%%%%%%%%%%%%%%%%%%%%%%%%%%%%%%%%%%%%%%%
%%%%%%%%%%%%%%%%%%%%%%%%%%%%%%%%%%%%%%%%%%%%%%%%%%%%%%%%%%%%%%%%%%%%%%%%%%%%%%%%
%%%%%%%%%%%%%%%%%%%%%%%%%%%%%%%%%%%%%%%%%%%%%%%%%%%%%%%%%%%%%%%%%%%%%%%%%%%%%%%%

\section{Introduction}
\label{sec:introduction}
Modern autonomous systems, such as self-driving cars, air taxis, and food delivery robots~\cite{waymo, volocity, serveRobotics2025}, are increasingly sophisticated, both in terms of the tasks they can accomplish and the hardware they use. This development underlines the need for powerful optimal control algorithms that can execute long-horizon tasks, planned over complex, and possibly data-driven, system models. In fact, optimal control has been used in a number of high-profile applications involving data-driven models, including quadrotor landing~\cite{shi2019neural}, quadruped control~\cite{xue2024full} and autonomous drifting~\cite{djeumou2023autonomous}.

Despite their great promise, however, optimal controllers are challenging to assure, which is an essential requirement for many of the safety-critical systems described above. In particular, complex tasks specifications and system dynamics result in large, non-convex optimization problems that are difficult to solve optimally, especially given computation constraints at runtime. As a practical solution, developers often utilize approximations such as gradient descent~\cite{shi2019neural} or sampling-based methods~\cite{xue2024full}, which provide fast performance at the expense of optimality.

In this paper, we aim to design a verification approach for autonomous systems with optimal controllers. As a first step, we focus on approaches based on gradient descent, which have proven to be very effective and are easier to analyze than more sophisticated techniques such as second-order or relaxation-based methods.

Verification of autonomous systems has received a lot of attention over the last two decades~\cite{althoff21}. A number of reachability approaches have been proposed, based on different computationally-convenient approximations such as
ellipsoids~\cite{botchkarev00}, polytopes~\cite{chutinan03}, zonotopes~\cite{althoff10} or Taylor models (TMs)~\cite{makino03,chen12}. Stochastic~\cite{abate2007computational} and simulation-based~\cite{duggirala15} verification methods have also been developed, as well as methods based on satisfiability modulo theories~\cite{gao13,kong15}. Furthermore, a range of dynamical system classes have been considered, including linear~\cite{frehse2011spaceex}, non-linear~\cite{althoff15}, hybrid~\cite{chen12} and transition systems~\cite{baier2008principles}.

In terms of control algorithms, a large class of methods is based on Hamilton-Jacobi reachability, which aims to design controllers which are safe by design~\cite{bansal2017hamilton}. Moreover, a lot of work has been done on verifying safety of neural network controllers~\cite{ivanov19,huang19,dutta19,sun19,tran19}, including the authors' Verisig approach~\cite{ivanov19,ivanov20,ivanov2020b,ivanov2021verisig}. However, we know of no prior work on reachability analysis of gradient-descent-based control for non-linear systems.

Our main idea is to treat gradient descent as a (digital) dynamical system, embedded in the original (physical) dynamical system. This results in an augmented system, which contains both the physical and the digital states, i.e.,~future controls, used in gradient descent. While this approach is intuitively appealing, it introduces significant challenges in terms of (1)~the dimensionality of the system, (2)~the verification horizon (essentially the product of physical control loops and digital gradient descent loops), and (3)~the complexity of the functions involved in gradient calculations.

To address the above challenges, we propose a reachability method based on TMs. To alleviate the high-dimensional state space challenge, we propose a digital TM shrink wrapping approach that aims to reduce the approximation error in digital TMs. Specifically, we decouple physical and digital TMs through a new shrink wrapping and unwrapping iteration that allows us to treat the problem as two lower-dimensional problems. Furthermore, we also introduce a method for symbolic propagation of remainder terms in digital states, which reduces error growth over long horizons and complex function approximations.

We evaluate the proposed method on two control platforms, a two-dimensional quadrotor and a cartpole. These systems illustrate different reachability challenges, such as higher control dimensions and longer planning horizons. In both cases, we design optimal controllers that aim to take the system to the origin and verify that these systems converge to a box around the origin, for a range of initial conditions.

In summary, this paper makes three contributions: 1)~a verification paradigm for autonomous systems with optimal controllers; 2)~a reachability analysis approach for gradient descent controllers based on TMs; 3)~two case study evaluations on systems with optimal controllers.

\section{Problem Statement} \label{sec:problem_statement}
We consider a discrete-time system\footnote{
Although the dynamics in the case studies are in continuous time we can view them as discrete-time systems because the controls are time-triggered.
} where the state ${{x}_k~ \in~\mathbb{R}^d}$ at step \(k\) is determined by a known state evolution function $f:\mathbb{R}^d\times\mathbb{R}^c\to \mathbb{R}^d$
\begin{align}
{x}_{k+1} &= f({x}_k, {u}_k) &
    {u}_k &= h({x}_k)
\label{eq:system}
\end{align}
where \(u_k \in \mathbb{R}^c\) is the vector of optimal control
over a cost function \(J: \mathbb{R}^{c\times H}\times\mathbb{R}^d\to\mathbb{R}\) up-to some horizon \(H\). The optimal control \(u_k\) is found by minimizing \(J\):
\begin{equation}
    \textbf{h}({x}_k) = \arg\min_{{u}_k:{u}_{k+H-1}} J({u}_k:{u}_{k+H-1},{x}_k) \label{eq:opt_j}
\end{equation}
where $\textbf{h}(x_k) = [u_k, u_{k+1}, \dots, u_{k+H-1}]$ consists of the entire horizon of controls, such that \(h(x_k) = \textbf{h}_1(x_k)\).
A larger horizon \(H\) allows the cost function to account for longer-term effects, at the cost of a higher-dimensional optimization problem.
Note that the optimal control at time \(k\) may be dependent on the controls and states up-to \(k +H\).
In this work, we consider \(J\) to be the quadratic cost function
\begin{equation}
J({u}_k:{u}_{k+H-1},{x}_k)
    = \sum_{i=0}^{H-1}
{x}_{k+1+i}^{\intercal}\mathcal{Q}{x}_{k+1+i} + {u}_{k+i}^\intercal\mathcal{R}{u}_{k+i}\label{eq:lqr_cost}
\end{equation}
where \({x}_{k+1}:{x}_{k+H}\) can be found by propagating controls through \eqref{eq:system}.~\footnote{The cost function may also have discounting where \(\mathcal{Q}\) and \(\mathcal{R}\) may not be the same matrix for all \(i\in[0,H-1]\). We also note that \(u_{k+H}\) cannot contribute to the final state \(x\).}

In general, the optimization problem \eqref{eq:opt_j} over the cost function \eqref{eq:lqr_cost} may not have a closed-form solution due to non-linearity or non-convexity in \(f\).
A common strategy then is to use an optimization algorithm such as gradient descent to get an approximate solution.
We consider the fixed step-size \(\alpha\) gradient descent algorithm to optimize \eqref{eq:opt_j} with \(T\) gradient updates.~\footnote{In our running examples, our models use a decaying \(\alpha\) learning rate. This is helpful for designing controllers that are easier to verify.}
We define \(z_i\) to be iterate $i$ of gradient descent for the controls \(\mathbf{h}(x_k)\).
The gradient update for \(z_{i+1}\) is
\begin{equation}\label{eq:grad_update}
    z_{i+1} = z_{i} - \alpha\nabla_{z_i} J (\cdot)
\end{equation}
where we initialize \(z_{0}=\mathbf{0}^{c\times H}\) and the final gradient step \(T\) results in the controls \(\mathbf{h}({x}_k)= z_{T}\).
We use zero initialization for consistency and ease of verification; further discussion on this choice is provided in Section~\ref{sec:conclusion}.
We now state the problem considered in this paper.
\begin{problem}
    Given an initial set of states \(\mathcal{X}_0\) and a dynamical system with optimal control, as defined in (\ref{eq:system}--\ref{eq:lqr_cost}), the problem is to
    find tight over-approximations for the reachable sets \(\mathcal{X}_1,\mathcal{X}_2,\dots,\mathcal{X}_K\) for all times $k \in \{1, \dots, K\}$.
\end{problem} 
\begin{remark} 
     Reachability analysis enables verification, as one could verify a safety property using intersection or union with the reachable sets. 
\end{remark}

\section{Background}
\label{sec:background}

The reachability approach described in this paper uses TMs as an approximation method~\cite{makino03}. This section provides the relevant background on TMs as it pertains to reachability analysis of dynamical systems.

At a high level, a TM for a function $f$ is a polynomial approximation of $f$ as well as worst-case error bounds over a domain of interest $D$. In particular, we say a polynomial~$p$ of degree $j$ is a $j$-degree polynomial approximation of $f$, around a point $x$, written $p(x) \equiv_j f(x)$, if all partial derivatives, up to order $j$, of $f$ and $p$ coincide at $x$. Furthermore, let~$\mathbb{I}$ be the set of all closed intervals $I = [a,b]$, with $a, b \in \mathbb{R}$. Then, a TM is defined as follows.

\begin{definition}[Taylor Model]
\label{def:tm}
Let $f: D \rightarrow \mathbb{R}$ be a $j$-times differentiable function of $n$
variables defined over a domain~$D \in \mathbb{I}^n$. Then a Taylor
model of $f$ over $D$ of degree $j$ is a pair $(p,I)$ of a polynomial
approximation $p$ and an error bound $I$ (also known as a remainder)
such that
\begin{align*}
1)&f(c) \equiv_j p(c), \text{where } c \text{ is the center of } D,\\
2)&\forall x \in D,\ f(x) \in \{p(x) + e \mid e \in I\}.
\end{align*}
\end{definition}

\begin{definition}[Taylor Model Range]
We define the range of $TM = (p,I)$ over domain $D$ as ${R(TM) = \{ p(x) + e \mid x \in D, e \in I\}}$.
\end{definition}

Thus, a TM can be thought of as a Taylor series approximation along with bounds on the remainder term. TMs are appealing due to their approximation power (given a high enough polynomial order) as well as their computational scalability on most systems considered in the literature. In particular, prior work~\cite{makino03} has introduced standard TM operations, including addition and multiplication, that allow one to perform arithmetic operations while preserving the conservative approximation of the resulting~TM.

\begin{definition}[Taylor Model Arithmetic]
\label{def:tma}
Consider two TMs, ${TM_1 = (p_1, I_1)}$ and ${TM_2 = (p_2, I_2)}$, defined
over a domain $D$. TM addition and multiplication are defined as
follows~\cite{makino03,chen13}:
\begin{align*}
TM_1 + TM_2 &= (p_1 + p_2, I_1 + I_2)\\
TM_1 \times TM_2 &= (p_1 \times p_2, I_M),
\end{align*}
where $I_M = \texttt{Int}(p_1)I_2 + \texttt{Int}(p_2)I_1 + I_1 \times I_2$, and $\texttt{Int}(p)$ is an interval bound of $p$ over $D$.
\end{definition}

The benefit of Definition~\ref{def:tma} is that TMs can now be used to conservatively approximate the reachable sets of dynamical systems such as the one in~\eqref{eq:system}, through repeated application of TM arithmetic. Furthermore, TMs are readily available for a number of analytic functions, such as exponentials and trigonometric functions~\cite{makino03}, which allows us to calculate reachable sets for a variety of systems encountered in practice. Finally, TMs can also be used for reachability analysis of continuous-time systems, e.g., through Picard iteration~\cite{chen13}, as well as for hybrid systems, e.g., by calculating intersections of reachable sets with guards and invariants using the concept of domain contraction~\cite{chen13}.

\textbf{TMs for dynamical systems.} For reachability analysis of dynamical systems such as the one in~\eqref{eq:system}, TMs are typically defined in terms of the initial conditions. Specifically, each polynomial is a function of the \emph{initial states} $x_0$ over domain $\mathcal{X}_0$. Thus, the reachable set at time step $k$ is represented as a vector of TMs, where each polynomial (of order $j$) has the form:
\begin{equation*}
    p_{x_{i,k}} = c + \sum_{n=1}^N c_{i,n} x_{1,0}^{q_{1n}}\cdots x_{d,0}^{q_{dn}},
\end{equation*}
where $N$ is the number of all monomials of order up to $j$, the $c$'s are constants, the $q_{mn} \in \{0, \dots, j\}$ are power terms satisfying $\sum_m q_{mn} \le j$ and $x_{i,0}$ denotes element $i$ of state vector $x_0$. Note that the domain is typically normalized to be $D = [-1,1]^d$, as that tends to result in better approximation performance~\cite{makino03}.

\textbf{TMs for non-polynomial functions.} For functions composed of multiple non-polynomial functions, e.g., $1/\cos^2(\theta)$, the standard approach is to start with TMs of known functions, e.g., cosine and division, and apply TM arithmetic to obtain a TM for the full function. Note that this approach may result in large approximation error as remainder terms are propagated through multiple TMs. We investigate methods to alleviate this challenge in Section~\ref{sec:technical}.

\section{Running Examples}
\label{sec:running_examples}
In this section we present two different dynamical systems that we use as running examples.
While the dynamics of both models are continuous, the control is time-triggered which allows us to treat them as discrete-time systems. 
This is possible since existing reachability tools can provide reachable sets at the time-triggered control steps.
Both systems have time-triggered control of 10Hz and the controllers are designed with the high-level objective to be within a small box of the origin. We aim to verify this by approximating the reachable sets up to five seconds of system execution time.
\subsection{Planar Quadcopter}
Consider a simplification of the common quadrotor model in two dimensions with a rigid body, named the planar quadcopter \cite[Chapter~3]{underactuated}. We modify the original model to also include saturation on controls as follows:
\vspace{-5pt}
\begin{align}
m\ddot{p}^h_{k} &= -F(\tanh(u^1_k) + \tanh(u^2_k))\sin\theta_k\\
m\ddot{p}^v_{k} &= F(\tanh(u^1_k) + \tanh(u^2_k))\cos\theta_k - mg\\
I\ddot{\theta}_k &= rF(\tanh(u^1_k)-\tanh(u^2_k))
\vspace{-5pt}
\end{align}
where \(m=0.486\unit{\kg}\) is the craft's mass, \(F=9.81\unit{N}\) is the maximum thrust of one propeller, \(g=9.81\unit{\m/\s^2}\) is the gravitational acceleration, \(I=0.00383\unit{\kg\cdot\m^2}\) is the moment of inertia, and \(r=0.25\unit{m}\) is the length
from a propeller to the axis of rotation.
\(p_k^h,p^v_k,\theta_k\) are the horizontal and vertical positions in space and the rotation.
Succinctly, we define the state and control at time \(k\) to be
\({
{x}_k=\begin{bmatrix}
    p^h &
    p^v &
    \theta &
    \dot{p}^h &
    \dot{p}^v &
    \dot{\theta}
\end{bmatrix}^\intercal,{u}_k=\begin{bmatrix}
    u^1&u^2
\end{bmatrix}^\intercal}
\).
  The predicted states are propagated in the cost function using an Euler approximation with a step size of \(0.1s\).
  
We use five gradient steps, and a decaying learning rate \(\alpha_i=0.095(0.51^i)\) for each gradient descent iterate \(i\in[0,4]\).
The finite cost-function horizon is \(H=4\). For the state and control costs in the cost \eqref{eq:lqr_cost}, set \({\mathcal{Q} = \mathop{\mathrm{diag}} ([
    \num{1.45e-01}, \num{1.95e+00}, \num{5e-03} ,\num{7.8e-03} ,\num{1.15e-03}, \num{1e-04}])}
\).
and \({\mathcal{R} = [\num{1.2e-04},\;\num{6.0e-05};\;
 \num{6.0e-05},\;\num{1.2e-04}]}\).
 The cost parameters were obtained through manual inspection so as to achieve the objective with minimum horizon length.

 To demonstrate the complexity of the gradients, for illustrative purposes, consider the first control in \(u_k\) for this problem. We write its gradient as a summation
 \vspace{-5pt}
 \begin{equation}
 \label{eq:sum_monomials_quadcopter}
 \vspace{-5pt}
     \frac{\partial{J}}{\partial u_{1,k}} = \sum_{i=1}^{75}g_{i}(u_k,x_k)
 \end{equation}
where each \(g_i(\cdot)\) is a non-linear term. The leading term \(g_1(\cdot)\) for \(u_{1,k}\) is
\({g_1(u_k,x_k) = C x_{k+1}\sin(\theta_k)\tanh({u_{k,1}})^2}\)
where \(C\) is a constant term consisting of a subset of the model parameters.
 In \(g_1(\cdot)\), there is a product of transcendental functions and the predicted states, which creates significant challenges for reachability analysis; other \(g_i(\cdot)\) terms can be similarly nonlinear and have even higher power.
\subsection{Cartpole}
The second case study is the cartpole system \cite{barto_neuronlike_1983}, consisting of a 
pole hinged to a cart on a frictionless track, with the goal of balancing the pole upright.
via horizontal forces. We use the corrected frictionless dynamics of Florian  
\cite{florian2007correct}, with state $x = [p,\ \dot{p},\ \theta,\ \dot{\theta}]^\top$ 
where $p$ is the cart position in meters, $\theta$ is the pole angle from vertical in radians, and $u$ is the applied horizontal force in Newtons.

Letting $\beta = (Fu + m_p l \dot{\theta}^2 \sin\theta)\,/\,(m_c + m_p)$, the dynamics are 
$\dot{p} = \dot{p}$, $\dot{\theta} = \dot{\theta}$,
$\ddot{p} = (\beta - m_p l \ddot{\theta} \cos\theta)/(m_c + m_p)$, and
\vspace{-5pt}
\begin{equation}
    \ddot{\theta} = \frac{g \sin\theta - \beta\cos\theta}{l \left( \frac{4}{3} - \frac{m_p \cos^2\theta}{m_c + m_p} \right)}
\vspace{-5pt}
\end{equation}
where \(m_c=1\unit{kg}\) is the cart mass, \(m_p=0.1\unit{kg}\) is the pole mass, \(l=0.5\unit{m}\) is the half-pole length, \(g=9.8\unit{m/\s^2}\) is gravitational acceleration, and \(F=10\) is the force scalar.

The cost function design follows the previous example. We use eight gradient steps and a decaying learning rate \({\alpha_i=0.01(0.8^i)}\) for each gradient step \({i \in [0,7]}\). The cost function horizon is \({H=6}\), using an Euler step of 0.2s. For the state costs,
\({\mathcal{Q}_k=1.1765^k \text{diag}([1.17647, 0, 1, 0])}\) for each predicted state in the cost function
{\(k \in [0, 5]\)} and {\(\mathcal{R}=[0.00075]\)}. All parameters were found using grid search over the parameter space.

Note that cartpole is more challenging to verify than the quadcopter, as it requires planning over a longer horizon to reach the origin and also has more complex dynamics. Again, consider the first control in \(u_k\) as a summation
\begin{equation}\label{eq:sum_monomials_cartpole}
\frac{\partial J}{\partial u_{1,k}} = 
    \left(\sum_{i=1}^{31158} n_i(u_k, x_k)\right) \bigg/ \left(\sum_{i=1}^{486} d_i(u_k, x_k)\right)\end{equation}
where each \(n_i(\cdot)\) and \(d_i(\cdot)\) are non-linear terms. The leading terms, \(n_1(\cdot)\) and \(d_1(\cdot)\) for \(u_{1,k}\) are
\vspace{-5px}
\begin{align}
\label{eq:cartpole_grad1}
   n_1(u_k, x_k) &= C_1 p_{k+6} u_{k+1} u_{k+3} \sin(\theta_{k+2}) \sin(\theta_{k+4}) \notag \\
                 &\quad \cdot \cos(\theta_{k+1}) \cos(\theta_{k+4}) \notag\\
   d_1(u_k, x_k) &= C_2 \cos(\theta_{k+1})^2
\vspace{-5px}
\end{align}
where \(C_1\) and \(C_2\) are constant terms consisting of the model parameters.
In \eqref{eq:cartpole_grad1}, the numerator and denominator each consist of products of predicted states, future controls, and transcendental functions.
The significant increase in terms and their complexity in comparison to the previous example reflects the compounding effect of a longer horizons and complicated dynamics.

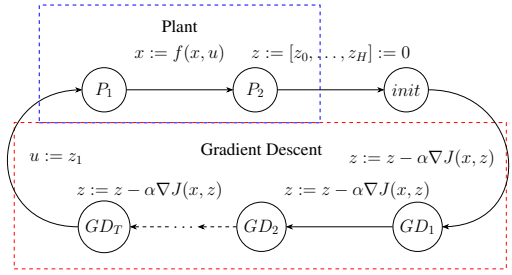
\begin{figure}
    \centering
    \begin{tikzpicture}[>=Stealth, node distance=2.5cm and 3cm, inner/.style={circle, draw, minimum size=1.25cm, font=\Large}]]
    % --- Main Locations ---
    % Anchor top-left coordinate for the entire structure
    \coordinate (topleft) at (0,0);

    % --- Plant Box (Top Row) ---
    % Define relative spacing and position for inner nodes
    \node[inner] (x)  {\(P_1\)};
    \node[inner, right=of x] (xp) {\(P_2\)};
    % Place the Plant title midway above the nodes
    \path (x) -- (xp) node[midway, yshift=1.8cm] (plantTitle) {Plant};
    % Define the dashed blue bounding box
    \node[draw=blue, dashed, inner xsep=1.2cm, inner ysep=0.3cm, fit=(x) (xp) (plantTitle)] (plantBox) {};

    % --- GradientDescent Box (Bottom Row) ---
    % Space the entire bottom box significantly below the top one
    \node[inner, below= of x] (zK)  {\(GD_T\)};
    % \node[right=of zK] (dots){\(\dots\)};
    \node[inner, right=of zK] (z2) {\(GD_2\)};
    \path (zK) -- (z2) node[midway] (dots) {\(\dots\)};
    \node[inner, right=of z2] (z1) {\(GD_1\)};
    \node[inner, right=of xp] (init) {\(init\)};
    % Place the GradientDescent title midway above the nodes
    \path (zK) -- (z1) node[midway, yshift=2.2cm] (gdTitle) {Gradient Descent};
    % Define the dashed red bounding box
    \node[draw=red, dashed, inner xsep=1.8cm, inner ysep=0.45cm, fit=(zK) (z1) (gdTitle)] (gdBox) {};

    % --- Internal Connections ---
    % Connections inside Plant box
    \path[->] (x) edge node[above, yshift=0.65cm] {\(x := f(x, u)\)} (xp);
    \path[->] (xp) edge node[above, yshift=0.65cm] {\(z := [z_0, \dots, z_H]:=0\)} (init);
    
    % Connections inside GradientDescent box
    \path[->] (z1) edge node[above, xshift=0.5cm, yshift=+0.65cm] {\(z := z - \alpha \nabla J(x, z)\)} (z2);

    % \path[->, dashed] (z2) edge node[above, yshift=0.7cm] {\(z := z - \alpha \nabla J(x, z)\)} (zK);
    \path[->, dashed] (z2) edge node[above, yshift=0.7cm] {} (dots);
    \path[->, dashed] (dots) edge node[above, yshift=0.65cm] {\(z := z - \alpha \nabla J(x, z)\)} (zK);
    % \path (zK) -- (z2) node[midway, yshift=+0.8cm] (unrolled) {};
    % --- External Curved Connections ---
    % Right side curve from xp down to z1
    \path[->] (init) edge[out=0, in=0, looseness=1.8]
        node[left, pos=0.5, xshift=-0.25cm, align=left] {\(z := z - \alpha \nabla J(x, z)\)}
        (z1);

    % Left side curve from zK up to x
    \path[->] (zK) edge[out=180, in=180, looseness=1.8]
        node[right, pos=0.5, xshift=+0.5cm, align=left] {\(u := z_{1}\)}
        (x);

    % --- Formula Text ---
    % Formula placed relative to the bottom box nodes
    % \node[anchor=north west, yshift=-0.8cm] at (z1.south west) (formula) {\(z_1 \leftarrow z_0 - \alpha \nabla J(x', z_0)\)};

    % --- Outer Border ---
    % After everything is placed, fit an outer node to the entire bounding box and draw it.
    % We use backgrounds library to ensure it's drawn behind everything if needed, 
    % though just a normal fit and draw works here.
    % \begin{scope}[on background layer]
    %     \node[draw, thin, inner sep=1cm, fit=(plantBox) (gdBox) (formula)] (outerBox) {};
    % \end{scope}

\end{tikzpicture}
    \vspace{-5pt}
    \caption{System dynamics as a transition system. The transition system states are shown with circles. The transition system is divided into two parts: the physical plant dynamics and the digital gradient descent. 
    % We initialize gradient descent with the zero-vector.
    }
    \label{fig:transition_system}
    \vspace{-7pt}
\end{figure}

\section{Approach Overview}
This section outlines our
approach to over-approximating the reachable sets of the full system by translating the system into a transition system with an augmented state. We begin by giving a formal description of the augmented system.

\subsection{Augmented Physical-Digital System}
In this section we define the augmented dynamical system's state which we use to compose gradient descent with the plant dynamics.  
To do this we consider the system in two parts: the physical dynamical system and the digital control system.
Intuitively, each gradient descent update could be considered as part of the digital dynamics update. We define the augmented system.

%Now we define the augmented state.
\begin{definition}[Augmented System]
\label{def:augmented_state}
 The augmented system consists of the physical state, \(x_k \in \mathbb{R}^d\), at time step \(k\) and the digital state, \(z_{i}^k\in\mathbb{R}^{c\times H}\), which denotes the matrix of future controls up-to the horizon \(H\) for iterate \(i\) of gradient descent at time \(k\). The gradient descent algorithm produces a sequence of iterates $z_1^k, z_2^k, \dots, z_T^k$. 
The iterates evolve according to the gradient update rule given in \eqref{eq:grad_update}.
% The augmented state at time \(k\) and gradient descent step \(i\) is defined as the tuple \((x_k,z_i^k)\).
\end{definition}
The augmented system allows us to track how gradient descent affects the approximation at each iterate rather than as a single transition black-box. Because the augmented system is a combination of both the physical and digital systems, we use a standard modeling framework for digital system model checking, namely transition systems~\cite{baier2008principles}.

\subsection{Augmented Transition System}
For verification purposes, we map the augmented dynamical system into an equivalent transition system.
We first provide a general definition of transition systems, followed by the augmented transition system considered in this paper.
\begin{definition}[Transition System]
\label{def:trans_sys}
A transition system is a tuple \((S,T)\) where \(S\) is a set of states and \(T\) is the set of transitions where \(T\subseteq S\times S\).
\end{definition}

\begin{definition}[Transition System for Augmented Physical-Digital System]
The transition system for the augmented dynamical system in Definition~\ref{def:augmented_state} is a transition system with states \(S =\mathcal{S} \times \mathcal{X} \times \mathcal{Z}\), where $\mathcal{S}$ are the discrete states \({\mathcal{S} = \{P_1,P_2,init, GD_1,\dots,GD_T\}}\), \(\mathcal{X}\) is the physical state space and \(\mathcal{Z}\) is the digital state space. The transitions are visualized in Figure~\ref{fig:transition_system}.
\end{definition}
The system cycles through a deterministic sequence of states: starting with the plant evolution (\(P_1 \to P_2\)), followed by zeroing of the gradient vector at (\(init\)), and a fixed sequence of gradient steps ($GD_1 \to \dots \to GD_K$) that determine the next control input.

Using this formulation, we can perform reachability analysis directly over the augmented transition system. To do this we use TMs to over-approximate the reachable sets efficiently.

\begin{algorithm}[!t]
\caption{Reachability Analysis of Control Systems with Optimal Controllers}
\label{alg:reachability_full}
\begin{algorithmic}[1]
\INPUT TM vector $TM_0$ for initial set $\mathcal{X}_0$, controller cost function $J$, physical horizon $K$, gradient horizon $T$, gradient step $\alpha$
   \State $TM_{x} \gets TM_0$
   \State $all\_reach\_sets \gets \{TM_x\}$
   \For{\textbf{each} $k$ in $\{1, \dots, K\}$}
   \State //gradient descent loop (Algorithm~\ref{alg:gd})
   \State $TM_{u} \gets reach\_gradient\_descent(TM_x, J, T, \alpha)$
   \State //physical system progress
   \State $TM_f \gets tm\_for\_dynamics(TM_x, TM_u)$
   \State $TM_x \gets TM_f \circ [TM_x, TM_u]$
   \State //physical state shrink wrapping
   \If{$large\_remainder(TM_x)$}
   \State $TM_x \gets shrink\_wrapping(TM_x)$
   \EndIf
   \State $all\_reach\_sets \gets all\_reach\_sets \cup TM_x$
   \EndFor
   \State \textbf{return} $all\_reach\_sets$
\end{algorithmic}
\end{algorithm}

\subsection{Taylor Models for Reachability}
The high-level reachability analysis algorithm is summarized in Algorithm~\ref{alg:reachability_full}.
Intuitively, we approximate each transition with a TM, as follows.
We can obtain a TM, \(TM_f\), for the dynamics using standard TM arithmetic, as described in Section~\ref{sec:background}. Obtaining a TM for the controls (in \(TM_u\)), as the output of gradient descent, is discussed Section~\ref{sec:technical}.

Algorithm~\ref{alg:reachability_full} is a standard reachability iteration for a dynamical system, with the exception that the controls are the output of a gradient descent algorithm, which requires the novel algorithm for reachability analysis of gradient descent.

Using these TMs, we can conservatively approximate the full optimization-in-the-loop system. While standard TM arithmetic is sufficient to approximate the physical plant dynamics \(TM_f\), the digital propagation of \(TM_u\) through gradient descent presents some considerable challenges.

\subsection{Technical Challenges}
While formulating the problem as an augmented transition system and approximating conservative bounds using TMs is intuitively appealing, finding the reachable sets is very challenging due to the coupled nature of the augmented system.
We discuss here three challenges that arise with verification over gradient descent.

\begin{itemize}
    \item \textbf{Dimensionality}: the augmented system has a much higher state dimension, proportional to the horizon of the cost function, than the physical system. This presents significant challenges for reachability methods, which often suffer from the curse of dimensionality.
    
    \item \textbf{Horizon length}:
    the augmented system has a significantly longer horizon than the original system -- the horizon is essentially the product of the physical system and the digital system horizons. In turn, longer horizons lead to compounding approximation error.
    
    \item \textbf{Gradient complexity}: computing the TM of the cost gradient may require repeated applications of TM  arithmetic to obtain the full function as evidenced in~\eqref{eq:sum_monomials_quadcopter} and~\eqref{eq:sum_monomials_cartpole}. This leads to a poor approximation due to growing remainder bounds.
\end{itemize}

In the following section we propose a TM-based reachability approach that alleviates these challenges, through using digital shrink wrapping and symbolic remainders.

\begin{figure}[t!]
  \centering
\includegraphics[width=0.7\linewidth]{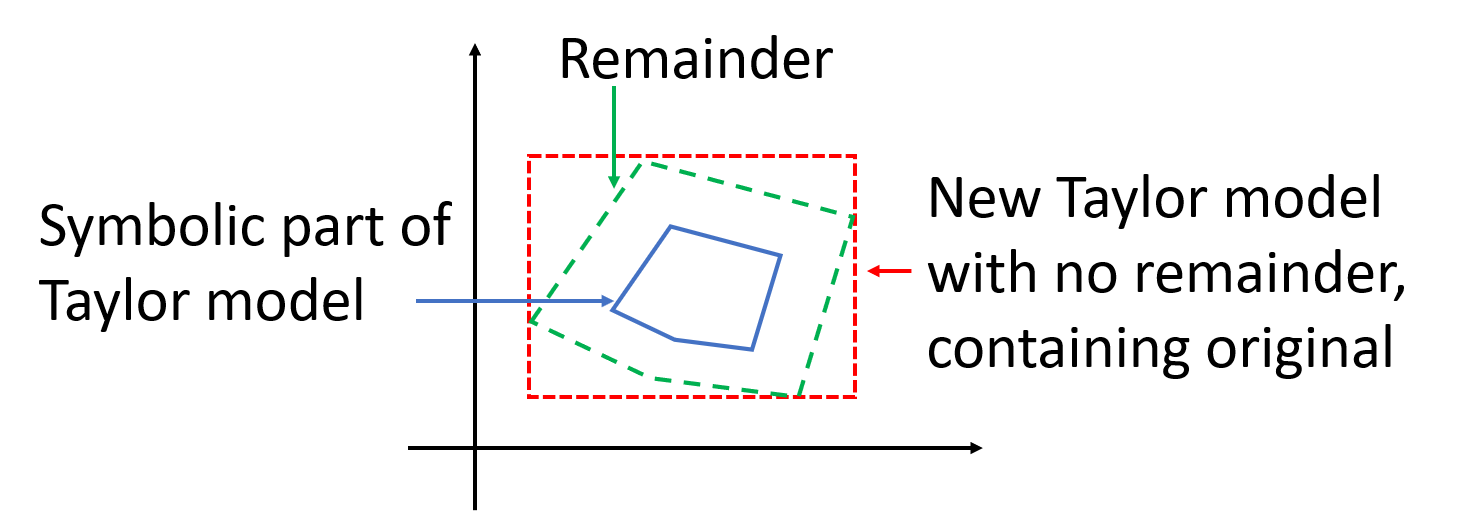}
\vspace{-5px}
  \caption{Illustration of shrink wrapping.}
  \label{fig:shrink_wrapping}  
  \vspace{-5px}
\end{figure}

\section{Reachability Analysis of Gradient Descent using Taylor Models}
\label{sec:technical}

This section presents the proposed approach for reachability analysis of dynamical systems with optimal controllers using TMs. As noted in the previous section, gradient descent poses significant challenges in terms of dimensionality, horizon length, and gradient complexity. This section describes the proposed techniques to alleviate these challenges, namely digital shrink wrapping and symbolic remainder propagation, as summarized in Algorithm~\ref{alg:gd}.

\begin{algorithm}[!t]
\caption{Reachability Analysis of Gradient Descent}
\label{alg:gd}
\begin{algorithmic}[1]
\INPUT TM vector $TM_x$ for physical states, controller cost function $J$, gradient horizon $T$, gradient step $\alpha$
    \State //initialize gradient descent variables to 0
    \State $TM_{z} \gets 0$
   \For{\textbf{each} $t$ in $\{1, \dots, T\}$}
   \State //TMs for predicted states $x_{k+i}$ in $J$
    \State $TM_{p} \gets tms\_predicted(J, TM_{x}, TM_{z})$
    \If{$large\_remainder(TM_{p})$}
   \State $TM_{p} \gets symbolic\_remainders(TM_{p})$
   \EndIf
    \State //Stack all TMs in augmented system state
    \State $TM_{full} \gets [TM_x, TM_{z}, TM_{p}$]
    \State //TM approximation of $\nabla J$
    \State $TM_{\nabla J} \gets tm\_for\_cost(TM_{full})$
   \State //Gradient update
   \State $TM_{z} \gets TM_{z} - \alpha TM_{\nabla J} \circ TM_{full}$
   \State //Digital shrink wrapping (and unwrapping)
   \If{$large\_remainder(TM_{z})$}
   \State $TM_{z} \gets shrink\_unwrapping(TM_{z})$
   \State $TM_{z} \gets digital\_shrink\_wrapping(TM_{z})$
   \EndIf
   \EndFor
   \State //Collect the first $c$ elements of gradient vector (i.e., $u_k$), to be applied (unwrapped) as control input to the plant
   \State $TM_u \gets shrink\_unwrapping(TM_{z,c})$
   \State \textbf{return} $TM_u$
\end{algorithmic}
\end{algorithm}

\subsection{Shrink Wrapping for Gradient Descent}
Shrink wrapping is a popular tool for reachability analysis using TMs as it provides great benefits for long-horizon problems~\cite{ivanov2021verisig,makino2005suppression,bunger2018shrink}. Shrink wrapping aims to slow down TM remainder growth: keeping the remainders small is essential since if the remainders grow large, then TM arithmetic devolves into interval analysis, and the approximation error begins to accumulate rapidly. As illustrated in Figure~\ref{fig:shrink_wrapping}, shrink wrapping works by replacing the large-remainder TM vector with a new TM vector that is a super set of the original one, but with no remainder. Although the new TM vector may temporarily result in increased approximation error, in the long run the error is reduced significantly since reachable sets are once again propagated symbolically.

In the literature, shrink wrapping approaches range from simple box approximations~\cite{ivanov2021verisig} (as shown in Figure~\ref{fig:shrink_wrapping}) to more sophisticated methods that obtain conservative new polynomials through bounding the Jacobian of the original TM~\cite{makino03,bunger2018shrink}. Unfortunately, existing techniques suffer from the curse of dimensionality -- the approximation error in the new TM vector increases dramatically for high-dimensional systems. Since the number of states in the augmented physical-digital system grows rapidly with the horizon length, effective reachability analysis now requires a new approach to alleviate the curse of dimensionality.

In this paper, we propose a modified shrink wrapping approach that performs separate shrink wrapping for physical and digital states. This approach not only alleviates the curse of dimensionality but it also results in faster computation as the plant polynomials only contain the physical states as variables, as opposed to the entire augmented vector of states. Despite its apparent simplicity, this approach requires careful consideration of the physical/digital TMs, since we are only refactoring part of the full TM vector at each time.

For simplicity, for the remainder of this paper we will assume shrink wrapping performs box shrink wrapping (as illustrated in Figure~\ref{fig:shrink_wrapping}), though we note that the proposed approach can be used with other existing methods as well.

\begin{definition}[Shrink Wrapping]
\label{def:box_shrink_wrapping}
Consider a TM vector, ${TMV = [TM_1, \dots, TM_d]}$, defined over variables\footnote{Although, as explained in Section~\ref{sec:background}, the TMs are defined over the states at time 0, $x_{i,0}$, for simplicity we drop the 0 subscript in this section.} $x_{1}, \dots, x_{d}$ and over domain $D = [-1,1]^d$. The TM vector ${TMV^{sw} = [TM_1^{sw}, \dots, TM_d^{sw}]}$ is a shrink wrapping of $TMV$ if
\vspace{-5px}
\begin{align*}
&\text{(i) } TMV^{sw} \text{ has no remainder, i.e., } TM_i^{sw} = (p_i^{sw}, 0), \forall i;\\
&\text{(ii) } R(TM_i) \subseteq R(TM_i^{sw}), \forall i.
\end{align*}
\end{definition}

\begin{proposition}[Box Shrink Wrapping~\cite{ivanov2021verisig}]
Consider a TM vector, $TMV = [TM_1, \dots, TM_d]$, defined over variables $x_{1}, \dots, x_{d}$ and over domain $D = [-1,1]^d$. Suppose $\texttt{Int}(TM_i) = [a_i, b_i]$ is an interval bound of each $TM_i$ over $D$. The TM vector ${TMV^{INT} = [TM_1^{INT}, \dots, TM_d^{INT}]}$ where $TM_i^{INT} = (p_i, 0)$, such that
\begin{align*}
    p_i = \frac{b_i + a_i}{2} + \frac{b_i - a_i}{2}x_{i},
\end{align*}
is a box shrink wrapping of $TMV$.
\end{proposition}

Intuitively, box shrink wrapping replaces all TMs with intervals, where each interval is represented as a linear function of the corresponding initial state dimension. As noted above, box shrink wrapping is only beneficial for low-dimensional systems as it tends to introduce significant error in high dimensions. To avoid performing box shrink wrapping on the augmented system, we propose to perform separate shrink wrapping for the physical and digital systems. At the same time, this step requires resolving the interdependencies between these subsystems, as explained next.

\begin{definition}[Digital Shrink Wrapping]
\label{def:digital_shrink_wrapping}
Consider a \emph{digital} TM vector, $TMV = [TM_1, \dots, TM_{c\times H}]$, defined over physical variables $x_{1}, \dots, x_{d}$ and over domain ${D = [-1,1]^d}$. The TM vector ${TMV^{dsw} = [TM_1^{dsw}, \dots, TM_{c\times H}^{dsw}]}$ is a digital shrink wrapping of $TMV$ if
\vspace{-5px}
\begin{align*}
&\text{(i) } TMV^{dsw} \text{ has no remainder, i.e., } TM_i^{dsw} = (p_i^{dsw}, 0), \forall i;\\
&\text{(ii) } R(TM_i) \subseteq R(TM_i^{dsw}), \forall i;\\
&\text{(iii) each } p_i^{dsw} \text{ may also contain variables } z_{1}, \dots, z_{c\times{H}}.
\end{align*}
\end{definition}

Note that digital shrink wrapping means we can use digital variables in the polynomials as well. This is needed because if the digital TMs are only functions of physical states, then we cannot shrink wrap the digital TMs without also shrink wrapping the physical TMs, as in doing so we would lose the dependencies between their corresponding reachable sets. Given this definition, we can now perform box shrink wrapping in a similar fashion to Definition~\ref{def:box_shrink_wrapping}.

\begin{proposition}[Box Digital Shrink Wrapping]
Consider a digital TM vector, $TMV = [TM_1, \dots, TM_{c\times H}]$, defined over variables $x_{1}, \dots, x_{d}$ and over domain ${D = [-1,1]^d}$, with $TM_i = (p_i, I_i)$, with $I_i = [a_i, b_i]$. The TM vector ${TMV^{dINT} = [TM_1^{dINT}, \dots, TM_{c\times H}^{dINT}]}$ where each ${TM_i^{dINT} = (p_i^{dINT},0)}$, such that
\vspace{-0pt}
\begin{align*}
    p_i^{dINT} = p_i + \frac{b_i + a_i}{2} + \frac{b_i - a_i}{2}z_i,
\end{align*}
is a digital shrink wrapping of $TMV$.
\end{proposition}

\begin{proof}
    Parts (i) and (iii) follow from the definition of ${TMV^{dINT}}$. To show part (ii), notice that the terms $(b_i + a_i)/2 + (b_i - a_i)z_{i}/2$ range between $[a_i, b_i]$ for $z_{i} \in [-1,1]$, which is exactly the range of the original remainder of $TM_i$. Thus, the range of $p_i^{dINT}$ is the same as the range $TM_i$.\qed
\end{proof}

Intuitively, box digital shrink wrapping moves each digital TM's remainder into a \emph{new} digital variable $z_i$. This means that in order to maintain the separation between physical and digital TMs, the TM for the output of gradient descent (i.e., the control TM), has to contain only physical variables. Thus, we propose the new concept of \emph{shrink unwrapping}, i.e., conservatively removing a variable from a TM.

\begin{definition}[Shrink Unwrapping]
Consider an augmented TM vector, ${TMV = [TM_1, \dots, TM_N]}$, defined over both physical and digital variables $x_{1}, \dots, x_{d}, z_{1}, \dots, z_{c\times H}$ and over domain ${D = [-1,1]^{d + c\times H}}$. The TM vector ${TMV^{su} = [TM_1^{su}, \dots, TM_N^{su}]}$, with $TM_i^{su} = (p_i^{su}, I_i^{su})$, is a shrink unwrapping of $TMV$ if
\vspace{-5px}
\begin{align*}
&\text{(i) } R(TM_i) \subseteq R(TM_i^{su}), \forall i;\\
&\text{(ii) each } p_i^{su} \text{ only contains variables } x_{1}, \dots, x_{d}.
\end{align*}
\end{definition}
\vspace{-5px}
As a first step, we propose a conservative shrink unwrapping method by bounding all monomial terms that depend on the digital variables and moving them into the remainder, as shown in the following proposition.

\begin{figure*}[t!]
    \centering
     \begin{subfigure}{0.49\textwidth}
 \centering
    \includegraphics[width=\linewidth]{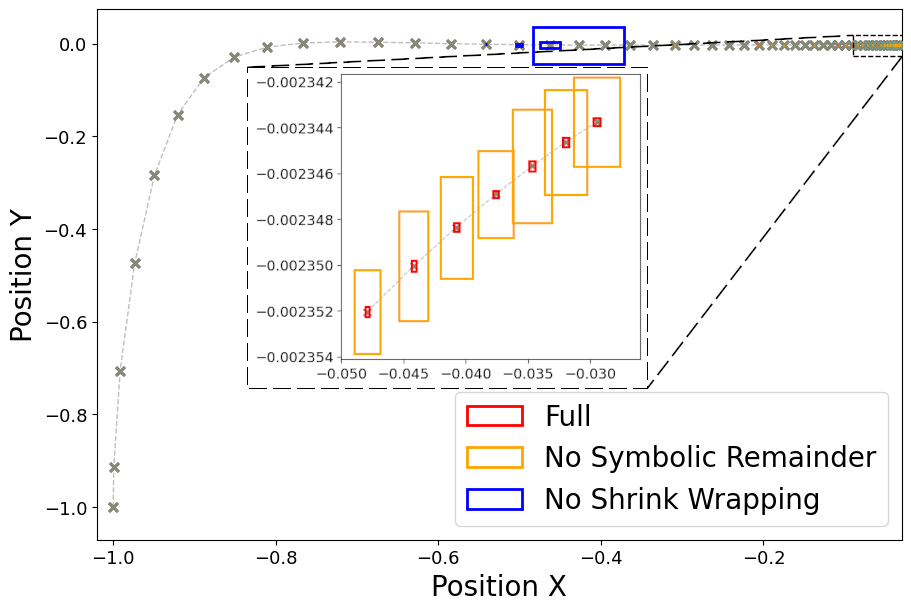}
     \caption{Planar Quadcopter Reachable Sets}
     \end{subfigure}
     \hfill
     \begin{subfigure}{0.49\textwidth}
    \centering
    \includegraphics[width=\linewidth]{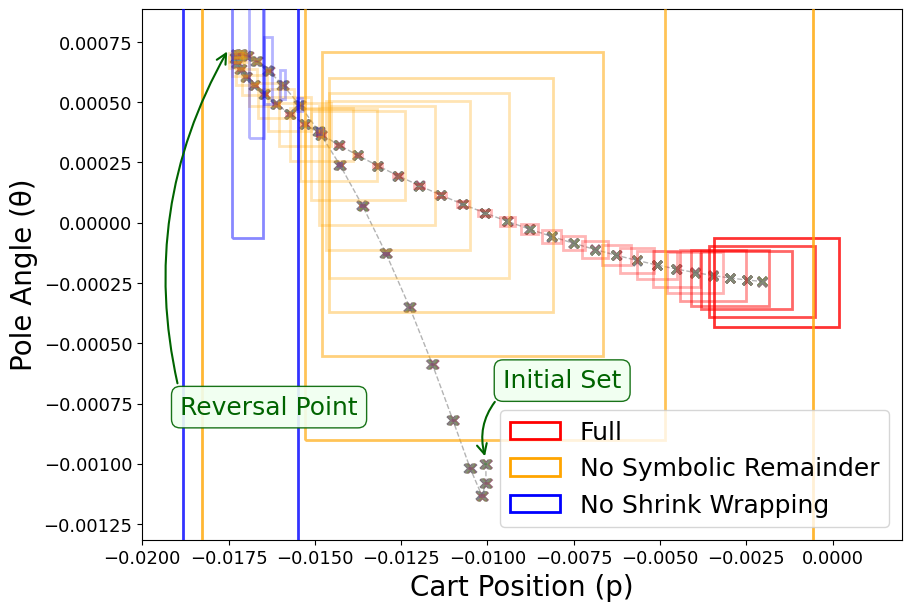}
    \caption{Cartpole Reachable Sets}
     \end{subfigure}
     \caption{Reachability analysis of the running examples. Reachable sets are shown in colored boxes. Simulation steps are plotted with X connected with a dashed gray line.
     In the planar quadcopter, the system starts in the bottom left and goes towards the origin. In cartpole, the system begins in the indicated initial set and first balances the pole before moving towards the origin. For clarity, later reachable sets in the cartpole figure have higher opacity. In both examples, using both of the techniques discussed greatly reduces the reachable set approximation error.  
     }
     \label{fig:reachable_sets}
     \vspace{-15pt}
\end{figure*}

\begin{proposition}[Box Shrink Unwrapping]
    Consider an augmented TM vector, $TMV = [TM_1, \dots, TM_N]$, defined over both physical and digital variables $x_{1}, \dots, x_{d}, z_{1}, \dots, z_{c\times H}$ and over domain ${D = [-1,1]^{d + c\times H}}$. Denote each TM as ${TM_i = (p_i, I_i)}$, where $p_i$ is decomposed as $p_i = p_i^x + p_i^{z}$, where $p_i^x$ is a function of only physical variables and $p_i^{z}$ is a function of both physical and digital variables. The TM vector ${TMV^{bsu} = [TM_1^{bsu}, \dots, TM_N^{bsu}]}$ where each $TM_i^{bsu} = (p_i^{bsu}, I_i^{bsu})$, such that
    \begin{align*}
    p_i^{bsu} &= p_i^x & 
    I_i^{bsu} &= I_i + \texttt{Int}(p_i^{z}),
\end{align*}
is a digital shrink wrapping of $TMV$.
\end{proposition}

\begin{proof}
    Part (ii) follows from the definition of $TMV^{bsu}$. To show part (i), notice that any value of $p_i^{z}$ is contained in $\texttt{Int}(p_i^{z})$, hence $R(TM_i) \subseteq R(TM_i^{bsu})$. \qed
\end{proof}

Thus, by performing both digital shrink wrapping and shrink unwrapping, we can ensure that the physical and digital states can be shrink wrapped separately. Finally, note that shrink unwrapping is also needed in between repeated digital shrink wrappings, to avoid adding a new set of digital variables after each gradient descent iteration.

\subsection{Symbolic Remainders for Predicted States}
In addition to the digital states $z$, we note that the gradient descent iteration also contains the predicted $x_{k+i}$ variables. While not directly modeled as states since they are functions of $x_k$ and $z_i^k$, they also greatly affect the approximation error since they essentially constitute a smaller reachability problem inside each gradient step.

To alleviate the error accumulation caused by predicted states, we propose the notion of symbolic remainders for digital variables. Propagating remainders symbolically is a well-known idea from the TM literature: at a high level, remainders are temporarily represented as variables that are propagated symbolically, thereby reducing the error caused by interval analysis~\cite{chen2015reachability}. In this paper, we build on this idea by augmenting the variable set even further, this time opportunistically including predicted states as described next.

Consider the TM for variable $x_{k+i}$, call it $TM_{px} = (p,I)$. If the magnitude of remainder $I$ exceeds a certain threshold, we reformulate $TM_{px}$ into a new TM by adding a new variable, $x_{k+i}$, whose range is exactly the magnitude of the remainder. While this idea is reminiscent of digital shrink wrapping in Definition~\ref{def:digital_shrink_wrapping}, the main difference is that the predicted variables serve as temporary remainder placeholders in case the remainders grow large. The symbolic remainders are shrink unwrapped from the digital states $z$ as well, so that the newly accumulated remainders can also be propagated symbolically.

\textbf{Implementation.} We implemented the proposed techniques as part of our Verisig tool~\cite{ivanov19}, which is based on the original TM implementation of the Flow* tool~\cite{chen13}. The augmented transition system is implemented as a general hybrid system, with no continuous dynamics in digital modes.

\section{Experiments}\label{sec:experiments}

In this section we demonstrate the reachable set computation for the running examples in Section~\ref{sec:running_examples}.

\textbf{Experiment Setup.}
In the planar quadcopter, we use the initial set \({(p_0^v,p_0^h)\in[-1\pm \num{1e-5} ]\times[-1 \pm \num{1e-5}]}\), with zero initial rotation and velocity. 
In the cartpole, we use the initial set \({p_0 \in [-0.0101, -0.01], \theta_0=-0.001}\), with zero velocity and angular velocity.
For both systems, we simulate 50 trajectories, sampled randomly from the initial set to provide an intuitive under-approximation of the reach sets.
For both examples, we evaluate the performance of the combined techniques of shrink wrapping and symbolic remainders, shrink wrapping without symbolic remainders, and symbolic remainders without shrink wrapping.
We note that because there has been no prior work on this problem, we evaluate different ablations of the proposed techniques in Section~\ref{sec:technical} and compare their importance for finding the reachable sets.
All experiments were run on a single core of an AMD EPYC 7413 24-Core Processor.

We can see in Figure~\ref{fig:reachable_sets}
that the reachable sets when using symbolic remainders and digital shrink wrapping remain the smallest throughout the trajectory. 
The figure also demonstrates the strength of our proposed techniques.
While in the planar quadcopter the reachable sets with no symbolic remainder remain relatively tight, cartpole demonstrates the importance of symbolic remainders in combination with shrink wrapping.
The reachability bounds of the system beyond the reversal point, where the velocity and angular velocity of the cartpole flip signs, diverge without shrink wrapping or symbolic remainders.
To describe the computational requirements of the proposed method, calculating reachable sets took 4237s for the planar quadcopter and 19570s for the cartpole.

From our experiments, it is evident that digital shrink wrapping has an immense effect on how conservative the reach sets become. Without shrink wrapping, verification would not be possible due to overly conservative approximation bounds.
Using symbolic remainders we can further tighten the bounds, even more so in systems with long horizons.
Together, symbolic remainders and digital shrink wrapping consistently produce tighter reachable sets across both systems.%, with shrink wrapping being the decisive factor in verification feasibility.

\section{Discussion and Conclusion}\label{sec:conclusion}
This paper presented a TM-based approach for reachability analysis of systems with optimal controllers using gradient descent. We modeled gradient descent as a digital dynamical system embedded in the physical one. To alleviate the challenges introduced by the increased dimensionality, horizon and complexity of the augmented system, we proposed modifications to existing TM techniques, namely digital shrink wrapping and symbolic propagation of digital remainders. We evaluated our method on two control systems, a two-dimensional quadrotor and a cartpole.

The complexity of this problem has revealed several directions for improvement in verification methodology. First of all, new verification techniques are needed to mitigate the curse of dimensionality, especially over long horizons. While shrink wrapping and symbolic remainders provide substantial benefits for the systems considered in this paper, new approaches are needed for more complex dynamics, such as data-driven models. Furthermore, there is a need for co-design of control and verification -- it is imperative to simplify both the cost function (e.g.,~through shorter horizons or fewer terms) and the optimization algorithm (e.g., through reducing the number of gradient steps and resetting gradient variables at the end of each control loop) so as to keep the approximation error low. Finally, it is important to consider other optimization strategies as well, such as sampling-based control and second-order methods.

\bibliography{root}
\bibliographystyle{IEEEtran}
\end{document}